\documentclass{tlp}

\usepackage{xcolor} 
\usepackage{amssymb}
\usepackage{hyperref}
\usepackage{algorithm}
\usepackage{algpseudocode}
\usepackage{commands}
\usepackage{alltt}
\newcommand{\csp}{\hspace{0.5cm}}
\newcommand{\teal}{\color{teal}}

\newcommand{\blue}{\color{blue}}

\newcommand{\impl}{\ {:\!-}\  } 

\newcommand{\normalized}{canonical\xspace}
\newcommand{\normalization}{canonicalization\xspace}
\newcommand{\Normalization}{Canonicalization\xspace}
\newcommand{\nrml}{canon\xspace}
\newcommand{\belief}{belief event\xspace}    
\newcommand{\beliefs}{belief events\xspace}  

\newcommand{\bprogram}{\cP^{\cB}}  
\newcommand{\bworlds}{W_{\pprog}^{\cD}}
\newcommand{\bworldmeas}{\beta}

\newcommand{\ourpabbv}{CaLP\xspace}
\newcommand{\ourpabbvs}{CaLPs\xspace}
\newcommand{\ourprog}{Capacity Logic Program}
\newcommand{\ourprogs}{Capacity Logic Programs}
\newcommand{\bmu}{\xi^{\cB}}
\newcommand{\Belief}{\mathit{Belief}}
\newcommand{\Plaus}{\mathit{Plaus}}
\newcommand{\mass}{\mathit{mass}}
\newcommand{\domain}{\mathit{domain}}
\newcommand{\facts}{\mathit{facts}}
\newcommand{\doms}{\mathit{doms}}
\newcommand{\comp}{\mathit{comp}}

\definecolor{yellow}{HTML}{e6ed2a}

\newtheorem{example}{Example}
\newtheorem{definition}{Definition}
\newtheorem{theorem}{Theorem}

\usepackage{listings}
\usepackage{amsmath}

\begin{document}

\lefttitle{Azzolini, Riguzzi, and Swift}

\jnlPage{1}{8}
\jnlDoiYr{2021}
\doival{10.1017/xxxxx}

\title[Integrating Belief Domains into Probabilistic Logic Programs] {Integrating Belief Domains into Probabilistic Logic Programs}

\begin{authgrp}
\author{\gn{Damiano} \sn{Azzolini}}
\affiliation{University of Ferrara}
\email{damiano.azzolini@unife.it}
\author{\gn{Fabrizio} \sn{Riguzzi} }
\affiliation{University of Ferrara}
\email{fabrizio.riguzzi@unife.it}
\author{\gn{Theresa} \sn{Swift}}
\affiliation{Coherent Knowledge, Inc.}
\email{theresasturn@gmail.com}
\end{authgrp}

\history{\sub{xx xx xxxx;} \rev{xx xx xxxx;} \acc{xx xx xxxx}}

\maketitle

\begin{abstract}
  Probabilistic Logic Programming (PLP) under the Distribution Semantics
is a leading approach to practical reasoning under uncertainty. An
advantage of the Distribution Semantics is its suitability for
implementation as a Prolog or Python library, available through two
well-maintained implementations, namely ProbLog and cplint/PITA.
However, current formulations of the Distribution Semantics use
point-probabilities, making it difficult to express epistemic
uncertainty, such as arises from, for example, hierarchical
classifications from computer vision models. Belief functions
generalize probability measures as non-additive capacities, and
address epistemic uncertainty via interval probabilities.  This paper
introduces interval-based {\em \ourprogs} based on an extension of the
Distribution Semantics to include belief functions, and describes
properties of the new framework that make it amenable to practical
applications.

\end{abstract}


\maketitle

\section{Introduction}
\label{sec:introduction}

Despite the importance of probabilistic reasoning in Symbolic AI, the
use of point probabilities in reasoning can be problematic.  Point
probabilities may implicitly assume a degree of knowledge about a
distribution that is not actually available, resulting in reasoning
based on unknown bias and variance.  Approaches to this lack of
knowledge, or {\em epistemic uncertainty}, often use probabilistic
intervals, such as Credal Networks~\citep{Cozm00} or interval
probabilities~\citep{Weich01}.  One appealing approach generalizes
probability theory through {\em belief functions} or {\em Dempster
  Shafer theory}~\citep{Shaf75}.

\begin{example}[Drawing Balls from an Urn]
\label{ex:urn}
  To illustrate how belief functions address epistemic uncertainty,
  suppose we want to determine the probability of picking a ball of a
  given color from an urn, given the knowledge:
  \begin{itemize}
  \item 30\% of the balls are red
  \item 10\% of the balls are blue
  \item 60\% of the balls are either blue or yellow
  \end{itemize}
  It is clear how to assign probability to red balls, but what about
  blue or yellow balls?  At least 10\% and at most 70\% of the balls
  are blue.  Similarly, between 0\% and 60\% of the balls are yellow.
  As will be explained in Section~\ref{sec:belief-prelim}, belief
  functions state that the {\em belief} that a ball {\em must} be blue
  is 10\%, while the {\em plausibility} that a ball {\em could} be
  blue is 70\%.
\end{example}
Example~\ref{ex:urn} describes a situation where {\em aleatory}
uncertainty, i.e., the uncertainty inherent in pulling a ball from the
urn, cannot be fully specified because of epistemic uncertainty, i.e.,
limited evidence about the distribution of blue and yellow balls.  Belief
functions separate aleatory from epistemic uncertainty by generalizing
probability distributions to {\em mass functions} that, in discrete
domains, may assign uncertainty mass to {\em any} set in a domain
rather than just to singleton sets.
For instance in Example~\ref{ex:urn}, a mass function $m_{urn}$ would
assign 0.6 to the set \texttt{\{blue,yellow\}}, 0.1 to
\texttt{\{blue\}}, and 0.3 to \texttt{\{red\}}.
Belief functions thus can be seen both as a way of handling epistemic
uncertainty or incomplete evidence~\citep{Shaf75}, and as a
generalization of probability theory \citep{HalF92}.
Mathematically, because belief functions are non-additive, they are
not measures like probabilities, but Choquet {\em
  capacities}~\citep{choquet1954theory}.

Let us now introduce a more involved example on how incorporation of
belief can help to address an important issue for neuro-symbolic
reasoning.

\begin{example}[Reasoning about Visual Classifications]
\label{example-ontology}
Consider the actions of a {\em unmanned aerial vehicle (UAV)} that
flies over roadways to monitor safety issues such as broken-down
vehicles on the side of a road, a task that includes visual
classification.
As a first step,
an image is sent to a neural visual model $V_{mod}$
that locates and classifies objects within bounding
boxes.  A reasoner then combines this visual information with traffic
reports and background knowledge to determine the proper action to
take.

To understand the relevance of belief functions to this scenario,
consider that visual models like $V_{mod}$ associate a confidence
score with each object classification.  Typically such scores are
normalized via a softmax transformation, but even if carefully
calibrated~\citep{GPSW17}, the scores may not have a frequentist or
subjectivist probabilistic interpretation. One way to provide a
frequentist interpretation is to construct a confusion matrix
$M_{Cls}$ from $V_{mod}$'s test set.  Given the set $Cls$ of object
types in $V_{mod}$, $M_{Cls}$ provides conditional probability
estimates about true classifications based on detected
classifications, and vice-versa.

A common situation when using a visual model is that some of its
training labels may be mapped to leaf classes of an ontology, while
others may be mapped to inner classes.
This is an example
of Hierarchical Multi-Label Classification~\citep{giunchiglia2020coherent}
where the classes are organized in a tree and the predictions must be coherent, i.e., 
an instance belonging to a class must belong to all of its ancestors in the hierarchy.
Figure~\ref{fig:conf-vector}
shows a fragment of an ontology and a thresholded confusion vector for
$V_{mod}$ representing $P(trueClass \mid detectedClass)$ for an unknown object
$O_{unk}$.  Note that $V_{mod}$  assigns probability mass both to
general levels, such as {\sf Passenger Vehicle}, and to leaf levels
such as {\sf Chevy}.  To reason about this situation, a flexible
model, capable of handling both aleatory and epistemic uncertainty, is
needed.

\begin{figure}[hbtp] 
\begin{small}
  \begin{tabular}{l|l}
\csp    Ontology                                & \csp Probability Mass Assignment \\ \hline
\csp {\blue Stationary Object}               &  \csp {\teal  0.0332} \\ 
\csp \csp  {\blue Dumpster}                  &  \csp {\teal  0.264} \\
\csp \csp {\blue Kiosk}                      &  \csp {\teal  0.011} \\
\csp {\blue Motorized Vehicle}               &  \csp {\teal  0.02} \\
\csp \csp {\blue Passenger Vehicle}         &  \csp {\teal  0.228}  \\ 
\csp \csp \csp {\blue Personal Passenger Vehicle} &  \csp {\teal  0.0489} \\ 
\csp \csp \csp \csp  {\blue Chevy }          & \csp {\teal 0.549} \\ 
\csp \csp \csp \csp {\blue Fiat}            & \csp {\teal  0.102} \\ 
\csp \csp {\blue Material Transport Vehicle} & \csp {\teal  0.203} \\
\csp \csp \csp {\blue Freight Truck}         & \csp {\teal  0.0182}       \\ 
\csp \csp \csp \csp {\blue Cement Truck}     & \csp {\teal  0.0709} \\ 
\csp \csp \csp \csp {\blue Cattle truck}     & \csp {\teal  0.0135}\\ 
\csp \csp {\blue \dots}                      &  \csp {\teal  \dots} \\
\end{tabular}
\end{small}
\caption{Example of Hierarchy of Objects.}
\label{fig:conf-vector}
\end{figure}
\end{example}

This paper shows how the distribution semantics for Probabilistic
Logic Programs
(PLPs)~\citep{DBLP:journals/ai/Poole93,DBLP:conf/iclp/Sato95} can be
extended to incorporate {\em belief domains} as a foundation for {\em
  \ourprogs (\ourpabbvs)}.  Specifically, our contribution
is to fully develop the distribution semantics with belief
domains, and based on that development,
\begin{itemize}
  \item We show that the set of all {\em belief worlds} for a
    \ourpabbv form a {\em normalized capacity} --- an analog of
    a probability measure that is used for belief functions.
  \item We demonstrate measure equivalence between
    (pairwise-incompatible) composite choices and belief worlds; and
  \item We show that a \ourpabbv also forms the basis of a
    normalized capacity space.
\end{itemize}
The upshot of these correspondences of belief worlds and composite
choices is that the implementation of \ourpabbvs will be possible for
systems based on the distribution semantics such as ProbLog and
Cplint/PITA.
The structure of the paper is as follows.  Section~\ref{sec:prelim}
presents background on belief functions and the distribution
semantics.  Section~\ref{sec:blps} extends the distribution semantics
to include belief functions, and characterizes its properties.
Section~\ref{sec:related} surveys related work and
Section~\ref{sec:conclusion} concludes the paper.

\section{Preliminaries} \label{sec:prelim}
Throughout this paper, we restrict our attention to discrete
probability measures and capacities, both defined over finite sets,
and to programs without function symbols.

\subsection{Capacities and Belief Functions} \label{sec:belief-prelim}
We recall the definition of probability measures and
spaces~\citep{ash2000probability}.

\begin{definition}[Probability Measure and Space] \label{def:prob-space}
Let $\cX$ be a non-empty finite set called the \textit{sample space}
and $\mathbb{P}(\cX)$ the powerset of $\cX$
called the {\em event  space}.
A {\em probability measure} is a function $P:\mathbb{P}({\cX})\rightarrow
\realnumbers$ such that
i) $\forall x \in \mathbb{P}(\cX), P(x) \geq 0$,
ii) $P(\cX) = 1$, and
iii) for disjoint $S_i \in \mathbb{P}(\cX): P(\bigcup_i \cS_i) = \sum_i P(\cS_i)$.
A {\em probability space} is a triple
$(\cX,\mathbb{P}(\cX),P)$.\footnote{Since we restrict our attention to
measures over finite sets, there is no loss of generality in the use of
powersets rather than $\sigma$-algebras, as any $\sigma$-algebra
over a finite set is isomorphic to the powerset of a set with smaller
cardinality.}
\end{definition}

All measures, including probability measures, have the additivity
property. That is, the measure assigned to the union of disjoint events is
equal to the sum of the measures of the sets.  Belief and plausibility
functions~\citep{Shaf75} are non-additive {\em capacities} rather than
measures~\citep{choquet1954theory,Grab2016}.

 \ourpabbvs in Section~\ref{sec:blps} use belief and plausibility
functions based on multiple capacity spaces or {\em domains}.
Accordingly, we use {\em domain identifiers} in the following
definitions.

\begin{definition}[Capacity Space]  \label{def:capacity}
Let $\cX$ be a non-empty finite set called the \textit{frame of
  discernment}. A {\em capacity} is a function $\xi:\mathbb{P}(\cX)
\rightarrow [0,1]$ such that i) $\xi(\emptyset)= 0$ and ii) for
$\cA,\cB \subseteq \cX$, $\cA \subseteq \cB \Rightarrow \xi(\cA) \leq
\xi(\cB)$.  If $\xi(\cX)=1$, $\xi$ is {\em normalized}.  A {\em
  capacity space} is a tuple $(D,\cX,\mathbb{P}(\cX),\xi)$ where $D$
is a string called a {\em domain identifier}.
\end{definition}

Belief and plausibility functions are capacities based on {\em mass
  functions}.

\begin{definition}[Mass Functions]  \label{def:mass-function}
Let $\cX$ be a non-empty finite set
and $D$ a domain identifier.  A {\em mass function} $\mass : D\times\mathbb{P}(\cX)
\rightarrow \realnumbers$
satisfies the following properties: i) $\mass(D,\emptyset) = 0$, ii) $\forall
x \in \mathbb{P}(\cX), \mass(D,x) \geq 0$, and iii)
$\sum_{x \in \mathbb{P}(\cX)} \mass(D,x) = 1$.
A set $x \subseteq \mathbb{P}(\cX)$ where $\mass(D,x) \neq 0$ is
a {\em focal set}.
\end{definition}

\begin{definition}[Belief and Plausibility Capacities] \label{def:belief-plaus}
Let $\cX$ be a non-empty finite set, $mass$ a mass function, $D$ a domain identifier, and $X \in
\mathbb{P}(\cX)$.
\begin{itemize}
\item The {\em belief} of $X$ is the sum of masses of all (not necessarily proper)
  subsets of $X$
  \begin{small}
    \begin{equation}
      \label{eq:belief}
      \Belief(D,X) = \sum_{B \subseteq X} \mass(D,B)
    \end{equation}
  \end{small}
    \item The {\em plausibility} of $X$ is the sum of the masses of
      all sets that are non-disjoint with $X$, and can be stated in
      terms of belief. Denoting the set complement of $X$ as $\neg X$:
      \begin{equation}
        \Plaus(D,X) = 1 - \Belief(D,\neg X)  
      \label{eq:plaus}
      \end{equation}
\end{itemize} 
\end{definition}
From the above definitions, it can be seen that $\Belief$ and $\Plaus$
are capacities.  Also, belief and plausibility are {\em conjuncts}, as
it is also true that $\Belief(D,A) = 1 - \Plaus(D,\neg A)$.
A capacity space $(D,\cX,\mathbb{P}(\cX),\xi)$ for which $\xi$ is 
a belief or plausibility function is a {\em belief domain}, and
$\mathbb{P}(\cX)$ is called the {\em belief event space} of such a space.


\begin{example} \label{ex:urn-domains}
We may represent a belief domain by the predicates {\em domain/2} which represents the frame of discernment ($\cX$) and {\em mass/3} which represents the mass function.
The belief domain of Example~\ref{ex:urn}, which we call {\em urn1}, can be represented as:
\begin{small}
\begin{equation*}
  \begin{split}
    & \domain(urn1,\{blue,red,yellow\}). \\
    & \mass(urn1,\{blue\},0.1). \ \mass(urn1,\{red\},0.3). \ \mass(urn1,\{blue,yellow\},0.6). 
  \end{split}
\end{equation*}
Here, $\Belief(urn1,\{red,yellow\}) = 0.3$ and $\Plaus(urn1,\{red,yellow\}) = 1 - \Belief(urn1,\{blue\}) = 0.9$.
\end{small}
\end{example}

\subsection{Probabilistic Logic Programs (PLPs)}
\label{secLplps}

Among several equivalent languages for PLP under the distribution semantics, we consider ProbLog~\citep{DBLP:conf/ijcai/RaedtKT07} for simplicity of exposition.
A \emph{ProbLog program} $\pprog=\tuple{\ruleset,\factset}$ consists
of a finite set $\ruleset$ of (certain) logic programming rules and a
finite set $\factset$ of \emph{probabilistic facts} of the form
$p_i::f_i$ where $p_i\in [0,1]$ and $f_i$ is an atom, meaning that we
have evidence of the truth of each ground instantiation $f_i\theta$ of
$f_i$ with probability $p_i$ and of its negation with probability
$1-p_i$.  Without loss of generality, we assume that atoms in
probabilistic facts do not unify with the head of any rule and that
all probabilistic facts are independent (cf.~\citep{Rig23-BK}).
Given a ProbLog program $\pprog=\tuple{\ruleset,\factset}$, its \emph{grounding} ($ground(\pprog)$) is defined as $\tuple{ground(\ruleset),ground(\factset)}$.
With a slight abuse of notation, we sometimes use $\factset$ to
indicate the set of atoms $f_i$ of probabilistic facts.  The meaning
of $\factset$ will be clear from the context.

\subsubsection{The Distribution Semantics for ProbLog Programs without Function Symbols}
\label{sec:pb-sem-no-fs}
For a ProbLog program $\pprog = \tuple{\ruleset,\factset}$, a
grounding $\tuple{\ruleset,\factset'}$, such that $\factset' \subseteq
ground(\factset)$ is called a {\em world}.  $W_{\pprog{}}$ denotes the
set of all worlds.  Whenever $\pprog$ contains no function symbols,
$ground(\factset)$ is finite, so $W_{\pprog{}}$ is also finite.
We also assume that each ProbLog program $\pprog =
\tuple{\ruleset,\factset}$ is {\em uniformly total}, i.e., that each
world has a total well-founded model~\citep{Przymusinski89}, so in each
world every ground atom is either true or false.
The condition that a
program is uniformly total is the most general way to express that
every world is associated with a stratified program, so that the
distribution semantics is well-defined.
We define a measure on a set of worlds, and how to compute the
probability for a query.

\begin{definition}[Measures on Worlds and Sets of Worlds]
\label{def:pworld-measure}
$\rho_{\pprog}:W_{\pprog{}}\rightarrow \realnumbers$ is a measure on worlds such that for each $w\in W_{\pprog{}}$
$$
\rho_{\pprog{}}(w) = \prod_{p::a \in \factset \mid a\in w} p \prod_{p::a \in \factset \mid a\not\in w} (1-p).
$$
$\mu_{\pprog{}}:\mathbb{P}(W_{\pprog{}})\rightarrow\realnumbers$ is a measure on sets of worlds as such that for each $\omega \in \mathbb{P}(W_{\pprog{}})$
$$
\mu_{\pprog{}}(\omega)=\sum_{w\in\omega}\rho_{\pprog{}}(w).
$$
\end{definition}
From Definition~\ref{def:pworld-measure}, $(W_{\pprog{}},\mathbb{P}(W_{\pprog{}}),\mu_{\pprog{}})$ is a probability space and $\mu_{\pprog{}}$ a probability measure.

\begin{definition}
\label{def:indicator_plp}
Given a ground atom $q$, define
$Q:W_{\pprog{}}\to\{0,1\}$ as
\begin{equation}
\label{random_var_finite_case_plp}
Q(w) = \left\{
                \begin{array}{ll}
                        1 & \mbox{if }w\models q \\
                        0 & \mbox{otherwise}
		\end{array}
                \right.
\end{equation}
 $w\models q$ means that $q$ is true in the well-founded model of $w$.
\end{definition}

Since $Q^{-1}(\gamma)\in \mathbb{P}(W_{\pprog{}})$ for all
$\gamma\subseteq \{0,1\}$, and since the range of $Q$ is measurable,
$Q$ is a random variable.  The distribution of $Q$ is defined by
$P(Q=1)$ ($P(Q=0)$ is given by $1-P(Q=1)$) and for a ground atom
$q$ we indicate $P(Q=1)$ by $P(q)$.

\noindent
Given this discussion, we compute the probability of a ground atom $q$ called \textit{query}
as
$$
P(q)    = \mu_{\pprog{}}(Q^{-1}(1)) 
        = \mu_{\pprog{}}(\{w \mid w \in W_{\pprog{}},w\models q\})
        = \sum_{w\in W_{\pprog{}} \mid w\models q}\rho_{\pprog{}}(w).
$$
That is, $P(q)$ is the sum of probabilities associated with each world in which $q$ is true.

\subsubsection{Computing Queries to Probabilistic Logic Programs}
\label{sec:plp-queries}

Let us introduce some terminology.  An \textit{atomic choice}
indicates whether or not a ground probabilistic fact $p_i :: f_i$ is
selected, and is represented by the pair $(f_i,k_i)$ where $k_i \in
\{0,1\}$: $k_i = 1$ indicates that $f_i$ is selected, $k_i = 0$
that it is not.
A set of atomic choices is \textit{consistent} if only one alternative
is selected for any probabilistic fact, i.e., the set does not contain
atomic choices $(f_i,0)$ and $(f_i,1)$ for any $f_i$.  A
\textit{composite choice} $\kappa$ is a consistent set of atomic
choices.

\begin{definition}[Measure of a Composite Choice]
\label{def:comp-choice-measure}
Given a composite choice $\kappa$, we define the function $\rho_{c}$ as
$$
\rho_{c}(\kappa)=\prod_{(f_i,1)\in\kappa}p_i\prod_{(f_i,0)\in\kappa}(1-p_i).
$$
\end{definition}
A \textit{selection} $\sigma$ (also called a {\em total composite
  choice}) contains one atomic choice for every probabilistic fact.
From the preceding discussion, it is immediate that a selection
$\sigma$ identifies a world $w_{\sigma}$.  The \emph{set of worlds
$\omega_\kappa$ compatible with a composite choice} $\kappa$ is
$\omega_\kappa=\{w_{\sigma}\in W_{\pprog} \mid
\kappa\subseteq\sigma\}$.  Therefore, a composite choice identifies a
set of worlds.
Given a set of composite choices $K$, the \emph{set of worlds $\omega_K$ compatible with} ${K}$ is
$\omega_{K}=\bigcup_{\kappa\in {K}}\omega_\kappa$.  
Two sets $K_1$ and $K_2$ of composite choices are
\emph{equivalent} if $\omega_{K_1}=\omega_{K_2}$, that is, if they identify the same set of worlds.
If the union of two composite choices $\kappa_1$ and $\kappa_2$ is not consistent, then $\kappa_1$ and $\kappa_2$ are \emph{incompatible}.
We define as \emph{pairwise incompatible} a set $K$ of
composite choices if $\forall \kappa_1\in K, \forall
\kappa_2\in K$, $\kappa_1\neq\kappa_2$ implies that $\kappa_1$
and $\kappa_2$ are incompatible.
If $K$ is a pairwise incompatible set 
of composite choices, define  $\mu_c(K)=\sum_{\kappa\in K}\rho_{c}(\kappa)$.

Given a general set $K$ of composite choices, we can construct a pairwise incompatible equivalent set through the technique of \emph{splitting}.
In detail, if $\f$ is a probabilistic fact and $\kappa$ is a composite choice that does not contain an atomic choice $(\f,k)$ for any $k$, the \emph{split} of $\kappa$ on $\f$ can be defined as the set of composite choices $S_{\kappa,\f}=\{\kappa\cup\{(\f,0)\},\kappa\cup\{(\f,1)\}\}$.
In this way, $\kappa$ and $S_{\kappa,\f}$ identify the same set of possible worlds, i.e., $\omega_\kappa=\omega_{S_{\kappa,\f}}$, and $S_{\kappa,\f}$ is pairwise incompatible.
It turns out that, given a set of composite choices, by repeatedly applying splitting it is possible to obtain an equivalent mutually incompatible set of composite choices.
\begin{theorem}[\citep{DBLP:journals/jlp/Poole00}] \label{incompatible_set}
        Let $K$ be a  set of composite choices.  Then there is
        a pairwise incompatible set of composite choices equivalent
        to $K$.
\end{theorem}

\begin{theorem} [\citep{DBLP:journals/ai/Poole93} Measure Equivalence
                 for Composite Choices ]\label{thm:pairwise-equiv}
	If $K_1$ and $K_2$ are both pairwise incompatible  sets of  composite choices such that they are equivalent, then $\mu_c(K_1)=\mu_c(K_2)$.
\end{theorem}

\begin{theorem}[Probability Space of a Program]
\label{omega_algebra}
Let $\pprog$ be a ProbLog program without function symbols and let  $\Omega_{\pprog}$ be $$\{\omega_K \mid K \mbox{ is a set of composite choices}\}.$$
Then $\Omega_{\pprog}=\mathbb{P}(W_{\pprog{}})$ and
$\mu_{\pprog}(\omega_K)=\mu_c(K')$ where $K'$ is a pairwise incompatible set of composite choices equivalent to $K$.
\end{theorem}
A composite choice $\kappa$ is an \emph{explanation} for a query $q$
if $ \forall w \in \omega_\kappa: w \models q$.
If the program is uniformly total, $w \models q$ is either true or
false for every world $w$.  Moreover, a set $K$ of composite choices
is \emph{covering} with respect to a query $q$ if every world in which
$q$ is true belongs to $\omega_K$.
Therefore, in order to compute the probability of a query $q$, we can
find a covering set $K$ of explanations for $q$, then make it mutually
incompatible and compute the probability from it. This is the approach
followed by ProbLog and PITA for example~\citep{DBLP:journals/tplp/KimmigDRCR11,RigS11a}.

\section{\ourprogs} \label{sec:blps}
  
We now present {\em \ourprogs\ (\ourpabbvs)} that extend PLPs to
include belief domains (Section~\ref{sec:belief-prelim}).
\begin{definition}[\ourpabbv rule] \label{def:calp_rule}
Let $\cD$ be a set of belief domains.  A {\em belief fact} has the
form $\mathit{belief}(D_k,B)$ where $D_k \in \cD$, and $B$ is an
element of the belief event space of $D_k$.  We call the value of the
first argument of a belief fact the {\em domain} and the value of the
second argument the {\em \belief}.\footnote{Alternatively, the belief
event could be called {\em evidence}.}

A \ourpabbv rule has the form $h
\impl l_0, \dots, l_m, p_0, \dots, p_n, b_0, \dots, b_o$ where $h$ is
an atom, each $l_i, i \in \{0, \dots, m\}$ is a literal, each $p_j, j
\in \{0, \dots, n\}$ is a probabilistic fact or its negation and each $b_k, k \in \{0, \dots,
o\}$ is a belief fact or its negation.
A negated belief fact 
$\text{\textbackslash +}\, \mathit{belief}(D,B)$ is equal to
$\mathit{belief}(D,\neg B).$
\end{definition}
Belief facts should be distinguished from the belief function {\em
  Belief/2} of Definition~\ref{def:belief-plaus}.  Note that belief
facts alone are sufficient to specify a program's use of a belief
domain, since plausibility and belief are conjuncts, so each can be
computed from the other.  We assume that all belief domains are
function-free.  To simplify presentation, we also assume belief facts
are ground.
\begin{definition}[\ourprog] \label{def:blp}
Let $\cF$ be a finite set of probabilistic facts and $\ruleset$ a
\ourpabbv ruleset.  Let $\cD$ be a finite set of belief domains and
$\cB$ the set of ground belief facts for all \beliefs of all
belief domains in $\cD$. The tuple $\pprog = (\ruleset,\cF,\cD,\cB)$
is called a {\em \ourprog}.  $\cB$ is called the {\em belief set} of
$\pprog$.
\end{definition}
As notation, for a set $\cB$ of belief facts and belief domain $D$,
$\facts(D,\cB)$ designates the set of belief facts in $\cB$ that have
domain $D$. In addition, $\doms(\cB)$ denotes the set of all domains that
are domain arguments of belief facts in $\cB$.

\noindent
When used in a belief world, a set of belief facts is made {\em \normalized} to take into account
that while belief facts from different belief domains are independent,
belief facts from the same belief domain are not.  Informally, if set of belief facts 
contains $\mathit{belief}(D_1,B_1)$ and $\mathit{belief}(D_1,B_2)$, these
belief facts are replaced by $\mathit{belief}(D_1,B_1 \cap B_2)$,
to ensure
all belief facts are independent.

\begin{definition} [\Normalization of a Set of Belief Facts] \label{def:belief-nrml}
The {\em \normalization of a set of belief facts $Bel$ } is
\[
  \mathit{\nrml}(Bel) = \{ \mathit{belief}(D,B) \mid  D \in \doms(Bel) \mbox{ and } B = \bigcap_{\mathit{belief}(D,B_i) \in Bel} B_i \}
  \]
  The set $\nrml(Bel)$ may contain facts whose \belief is $\emptyset$.  If
  $\nrml(\cB)$ contains such facts it is {\em inconsistent}, otherwise
  it is {\em consistent}.
\end{definition}
\begin{example}[\Normalization of Sets of Belief Facts]
  To further motivate
  \normalization,
consider the belief domain {\em urn1} of Example~\ref{ex:urn-domains}
and
consider a belief set containing {\em belief(urn1,\{blue\})} and {\em
  belief(urn1,\{blue,yellow\})}.  From an evidentiary perspective
(cf.~\citep{Shaf75}) the belief set contains evidence both that a blue
ball was chosen and that a blue {\em or} yellow ball was chosen.
Clearly, the evidence of a blue ball also provides evidence that a
blue {\em or} yellow ball was chosen.  \Normalization takes the
intersection of these two belief events, retaining only {\em
  belief(urn1,\{blue\})}.
Alternatively, consider a belief set $Bel_1$ that contains both {\em
  belief(urn1,\{blue\})} and {\em belief(urn1,\{red\})}.  Since a ball
drawn from an urn cannot be both blue and red, the belief of
$\mathit{belief}(urn1,(\{blue\}\cap\{red\}))$ is 0 and $\nrml(Bel_1)$ is
inconsistent.
\end{example}
\begin{definition} [Belief World]
\label{def:belief_worlds}
For a ground \ourpabbv $\pprog = \tuple{\ruleset,\factset,\cD,\cB}$, a
{\em belief world} $w = \tuple{\ruleset,\factset',\cD,\cB'}$ is such
that $\factset' \subseteq \factset$, and $\cB' \subseteq \cB$ is a
consistent \normalized belief set that contains a belief fact for every
domain in $\cD$.  $\bworlds$ denotes the set of all belief worlds.
\end{definition}
\noindent
Since $\cD$ is a finite set of function-free belief domains,
$\bworlds$ is finite.  Also, for any world
$(\ruleset,\factset',\cD,\cB') \in \bworlds$, $\cB$ contains exactly
one belief fact for every domain in $\cD$: by
Definition~\ref{def:belief_worlds}, $\cB'$ contains at least one belief
fact for a given $D_i \in \cD$, while \normalization ensures that there is
at most one belief fact for $D_i$.

As defined in Section~\ref{secLplps}, CaLPs are assumed to be
uniformly total: that each world has a two-valued well-founded model
(i.e., that each world gives rise to a stratified program).
Thus, to extend the semantics of Section~\ref{secLplps} it needs to be
ensured that the well-founded model can be properly computed when
belief facts are present, and for this the step of {\em completion} is
used.
Suppose $w$'s belief set contained {\em belief(urn1,\{blue\})} as the
(unique) belief fact for the belief domain {\em urn1} of Example~\ref{ex:urn-domains}.
Because $\mathit{belief}(urn1,\{blue\})$ provides evidence for
$\mathit{belief}(urn1,\{blue,yellow\})$, a positive literal
$\mathit{belief}(urn1,\{blue,yellow\})$ in a rule should succeed.

\begin{definition}[Well-Founded Models for \ourpabbv Programs]\label{def:belief-wfm}
Let $w = (\ruleset,\cF,\cD,\cB)$ be a belief world.  The {\em
  completion} of $\cB$, $\comp(\cB)$ is the smallest set such that
if $\mathit{belief}(\cD,B_1) \in \cB$, and $B_2$ is an element of the
belief event space
of $D$ such that $B_1 \subseteq B_2$, then $B_2 \in comp(\cB)$.
The well-founded model of $w$, $\mathit{WFM}(w)$ is the well-founded
model of $\ruleset \cup \cF \cup comp(\cB)$.  For a ground atom $q$,
$w \models q$ if $q \in \mathit{WFM}(w)$.
\end{definition}
Thus a belief set $\cB$ must be \normalized before being used to
construct a world $w$.  On the other hand, the completion of $\cB$ is used when computing
$WFM(w)$ and ensures rules requiring weak evidence (such as that a ball is
yellow or blue) will be satisfied in a belief world that provides
stronger evidence (such that a ball is blue).

Belief domains can be seen as a basis for
imprecise probability~\citep{HalF92}, so the capacity of a belief
world is
an interval containing both belief and plausibility.
Accordingly, circumflexed products and sums represent interval
multiplication and addition.  In the following definition,
$\beta_{prb}$ computes the portion of the capacity due to
probabilistic facts (similar to Definition~\ref{def:pworld-measure}),
and $\beta_{blf}$ the portion due to belief facts.
\begin{definition}[Capacity of a Belief World] \label{def:belief-world-measure}
$\beta_{prb},\beta_{blf}, \beta:\bworlds \rightarrow \realnumbers^2$
  are capacities such that for $w = (\ruleset,\cF,\cD,\cB)$ 
  \begin{equation*}
    \begin{split}
      & \beta_{prb}(w)= \widehat{\prod}_{p::a\in w}[p,p] \widehat{\prod}_{p::a\not\in w}[(1-p),(1-p)] \\
      & \beta_{blf}(w)= \widehat{\prod}_{\mathit{belief}(D,B)\in \cB} [\Belief(D,B),Plaus(D,B)] \\
      & \beta(w)= \beta_{prb}(w)\widehat{\times}\beta_{blf}(w)
    \end{split}
  \end{equation*}
\end{definition}
The restriction that a world's belief set contain a belief fact for
each belief domain need not affect the capacity of a world, since the
fact may indicate trivial evidence for the entire discernment frame,
for which the belief and plausibility are both 1.  For
instance, {\em \{blue, red, yellow\}} in the $urn1$ domain has a
belief of 1.

To assign a non-additive capacity to a {\em set} of belief worlds, the
super-additivity of belief and plausibility capacities must be
addressed~\footnote{A capacity $\xi$ is super-additive if for two sets
$A,B$ $\xi(A\cup B)$ may be greater than $\xi(A)+\xi(B)$.}. The
following definition provides one step in this process.
\begin{definition} \label{def:bdp}
  Let $\cS$ be a set of belief facts, and $D \in \doms(\cS)$.
  The {\em upper belief domain probability} is
  \begin{small}
    \begin{tabbing}          
      foooooooooooo\=fofooooooooofoooooooofoooooooooooooo\=ooo\=\kill
      $BP^{\uparrow}(D,\cS) =$
      \> $[1,1]$ \> if $\facts(D,\cS)$ is empty; and \\
      \> $[\Belief(D,\bigcup B_i)$,$\Plaus(D,\bigcup B_i)]$
      \> for $B_i \in \facts(D,\cS)$ otherwise
    \end{tabbing}
  \end{small}
\end{definition}

Next, a set $\cB$ of belief worlds is partitioned by grouping into
cells worlds in $\cB$ that have the same set of probabilistic facts.
Within each cell, the probabilistic facts are factored out and the
belief events for each domain are then unioned, leading to a single
world and unique capacity for each partition cell.  Finally the
capacities of all cells are added.  This process is necessary to avoid
counting the mass of probabilistic facts
more than once when the capacities of belief worlds are summed.

\begin{definition}[Capacities of Sets of Belief Worlds] \label{def-bel-meas-set}
  Let $\cP = (\ruleset,\cF,\cD,\cB)$.  The partition function
  $Dom{Prtn}(\cS):\mathbb{P}(\bworlds) \rightarrow
  \mathbb{P}(\mathbb{P}(\bworlds))$, given $\cS \in
  \mathbb{P}(\bworlds)$ produces the {\em domain partition} of $\cS$:
  the coarsest partition such that for each $S_i \in Dom{Prtn}(\cS)$:
  \[(\ruleset,\cF_j,\cD,\cB_j) \in S_i \mbox{ and }(\ruleset,\cF_k,\cD,\cB_k)
  \in S_i \Rightarrow \cF_j = \cF_k.\]
  Using the domain partition, we define
  $probfacts(S_i)$ as the unique $\cF'$ for any $w =
  (\ruleset,\cF',\cD,\cB')$ s.t., $w \in S_i$.  The capacity
  $\beta_{prtn}$ for $S_i \in Dom{Prtn}(\cS)$ is
  \begin{equation}
  \bworldmeas_{prtn}(S_i) = \beta_{prb}(probfacts(S_i)) \widehat{\times}
  \widehat{\prod}_{(\ruleset,\cF,\cD,\cB) \in \cS_i;D
    \in doms(\cB)}BP^{\uparrow}(D,\cB).
  \end{equation}
  Finally, the capacity for $\cS$ is
  \begin{equation}
  \bmu(\cS) = \widehat{\sum}_{S_i \in Dom{Prtn}(\cS)} \bworldmeas_{prtn}(S_i).
  \end{equation}\label{eq:bprtn}
\end{definition}

\begin{example} [Computing the Capacity for a Set of Belief Worlds]\label{ex:urn2}
  Consider a new domain $urn2$ analogous to $urn1$ of
  Example~\ref{ex:urn-domains}:
    \begin{it} 
\begin{tabbing}
       foooooooooooofofoooooooooooo\=foooooooofooooooooooooooooo\=ooo\=\kill
       domain(urn2,\{green,orange,purple\}) \\
       mass(urn2,\{green\},0.1)\>mass(urn2,\{orange\},0.3)\>mass(urn2,\{green,purple\},0.6)
     \end{tabbing}
    \end{it}

  Let \ourpabbv $\cP = (\ruleset,\cF,\cD,\cB)$ where $\cF =
  \{p_1::f_1,p_2::f_2\}$ and $\cD = \{urn1,urn2\}$, along with a set $\cS$ of
  3 worlds of $\cP$ with the following fact sets and belief domains.
\begin{small}
  \begin{tabbing}
   foooooooofooooooooofoooooooofoooooooooooooo\=ooo\=\kill
   $w_1: \cF_1 = \{p_1\}, \cB_1= \{\mathit{belief}(urn1,\{blue\})\}$ \>  $w_2: \cF_2 = \{p_2\}, \cB_2= \{\mathit{belief}(urn2,green)\}$\\
        $w_3: \cF_3 = \{p_1\}, \cB_3= \{\mathit{belief}(urn1,\{yellow\})\}$
  \end{tabbing}
\end{small}
Following Definition~\ref{def-bel-meas-set}, $DomPrtn(\cS$) partitions
$\cS$ into $\{w_1,w_3\}$ and $\{w_2\}$.  To compute
$\beta_{prtn}(\{w_1,w_3\})$, $BP^{\uparrow}$ is calculated for each
belief domain giving 
\begin{small}
  \[[\mathit{Belief}(urn1,\{blue,yellow\}),\mathit{Plaus}(urn1,\{blue,yellow\}])= [0.7,0.7]\]
\end{small}
for the domain {\em urn1} and $[1,1]$ for the domain $urn2$.
Next, $\beta_{prb}(\{w_1,w_3\})$ is calculated as
$[p_1,p_1]\widehat{\times}[(1-p_2),(1-p_2)]$ and multiplied with
$[0.7,0.7]$.  The result is then summed with $\beta_{prtn}(\{w_2\})$.

\end{example}

One might ask why Definition~\ref{def-bel-meas-set} uses
$BP^{\uparrow}$ to combine beliefs of the same domain rather than
Dempster's Rule of Combination.  Dempster's rule was designed to
combine independent evidence to determine a precise combined belief.
This approach is useful for some purposes but has been criticized for
many others (cf.~\citep{Josa16}).  When combining \ourpabbv worlds,
epistemic commitment for a given domain is reduced by $BP^{\uparrow}$
in a manner similar to combining \ourpabbv explanations
(Section~\ref{sec:blp-queries}).

 Since the range of $\bmu$ is $\mathbb{R}^2$, we designate $proj_{B},
 proj_{P}:\mathbb{R}^2 \rightarrow \mathbb{R}$ as projection functions
 for the belief and plausibility portions respectively.
The following theorem states that capacity spaces formed using belief
worlds and the capacity function $\bmu$ have belief and plausibility
capacities whose value is 1 for the entire belief event space.
\begin{theorem}
  Let $\cP = (\ruleset,\factset,\cD,\cB)$ be a \ourpabbv.
  Then, 
  $(\bworlds,\mathbb{P}(\bworlds),proj_{B}(\bmu))$ and
  $(\bworlds,\mathbb{P}(\bworlds),proj_{P}(\bmu))$ are normalized
  capacity spaces.
\end{theorem}
\begin{proof}
   From Definitions \ref{def:capacity} and \ref{def-bel-meas-set} it
   is clear that $\bmu$ is a capacity.  To show that $\bmu$ is
   normalized, consider that $\bmu$ is a function composition, where
   the first function, $DomPrtn(\bworlds)$ returns the domain
   partition
   of $\bworlds$, while $\beta_{prb}(probfacts(S_i))$ is the separate
   probability of $\cF$.  Consider a cell $\cS_i \in DomPrtn(\cS)$.
   From Eq. (4)
  \[\bworldmeas_{prtn}(S_i) = \beta_{prb}(probfacts(S_i)) \widehat{\times}
  \widehat{\prod}_{(\ruleset,\cF,\cB) \in \cS_i;D \in doms(\cB)}BP^{\uparrow}(\cD,B)\]
   Every world in $\cS_i$ contains the same set of probabilistic
   facts, $(probfacts(S_i))$, but differs in belief facts.  Let $D$ be
   a belief domain in $\cD$.  Since $\cS_i$ is a partition of
   $\bworlds$, and because all belief domains in all worlds are
   normalized, $S_i$ contains a world with every \belief of $D$.
   Thus $BP^{\uparrow}(D,B) = [1,1]$
   so that
   \[\bworldmeas_{prtn}(S_i) = \beta_{prb}(probfacts(S_i)) \widehat{\times}
   \widehat{\prod}_{(\ruleset,\cF,\cB) \in \cS_i;D \in doms(\cB)}[1,1]
   = \beta_{prb}(probfacts(S_i)) \]  
     Thus, the value of $\bworldmeas_{prtn}(\cS_i)$ relies only on the
     probabilistic facts that were used for the partition.  Since
     $\bworlds$ contains all possible worlds, $\bmu(\bworlds)$ reduces
     to the sum of all subsets of probabilistic facts of $\cP$, which
     is 1.
\end{proof}

\begin{definition}[cf. Section~\ref{secLplps}, Definition~\ref{def:indicator_plp}]
  \label{def:indicator}
        Given a ground atom $q$ in $ground(\pprog)$
	we define  $Q_{\cB}:  \bworlds\to\{0,1\}$ as
        \begin{small}
	\begin{equation}
	  Q_{\cB}(w)=\left\{\begin{array}{ll}1 & \mbox{if }w\vDash q\ 
          (\mathit{Definition}~\ref{def:belief-wfm})\\
			0     & \mbox{otherwise}
		\end{array}\right.\label{random_var_finite_case}
	\end{equation}
        \end{small}
\end{definition}
By definition
$Q_{\cB}^{-1}(\gamma)\in \mathbb{P}(\bworlds)$
for $\gamma\subseteq \{0,1\}$, and since the range of $Q_{\cB}$ is
measurable, $Q_{\cB}$ is a capacity-based random variable.

Given the above discussion, we compute $[\mathit{belief}(q),\mathit{plaus}(q)]$ for a
ground atom $q$ as
\begin{small}
  \[[\mathit{belief}(q),\mathit{plaus}(q)] = \bmu(Q_{\cB}^{-1}(1))=\bmu(\{w \mid w\in\bworlds, w\models q\}).\]
\end{small}

\section{Computing Queries to \ourprogs} \label{sec:blp-queries}

We extend the definition of atomic choice from
Section~\ref{sec:plp-queries} to \ourpabbvs as follows.
\begin{definition}[Capacity Atomic Choice (\ourpabbvs)]
  For a \ourpabbv $\bprogram= (\ruleset,\cF,\cD,\cB)$ a
  \textit{capacity atomic choice} is either:
  \begin{itemize}
  \item An {\em atomic Bayesian choice} of ground probabilistic fact $p :: f$,
    represented by the pair $(f,k)$ where $k \in \{0,1\}$. $k = 1$
    indicates that $f$ is selected, $k = 0$ that it is not.
    
  \item An {\em atomic belief choice} of a ground belief fact {\em
    belief($D$,B)} where $D$ is a belief domain in $\cD$ and $B$ is a
    \belief in $D$.  The choice is represented as {\em belief(D,B)}.
  \end{itemize}
\end{definition}
Atomic belief choices with different belief domains are considered
  independent choices, much as atomic Bayesian choices.  However,
  atomic belief choices with the same belief domain are {\em not}
  independent: so fitting them into the distribution semantics
  requires care.

Given a set $\cS$ of atomic choices, let $Bay(\cS)$ be the subset of
$\cS$ that contains only the atomic Bayesian choices of $\cS$ and let
$Bel(\cS)$ be the subset of $\cS$ that contains only the atomic belief
choices of $\cS$. Note that $Bel(\cS)$ is a set of belief facts,
so that definitions of Section~\ref{sec:blps} can be used directly
or with minimal change. Then $\cS$ is consistent if both $Bay(\cS)$
and $Bel(\cS)$ are consistent (cf. Definition~\ref{def:belief-nrml}).
A
 {\em capacity composite choice $\kappa$} is a consistent set of
atomic choices for a \ourpabbv $\bprogram= (\cP,\factset,\cD,\cB)$.
Given a capacity composite choice $\kappa$, let $doms(\kappa)$ be
$\{D \mid \mathit{belief}(D,B)\in \kappa\}$.
\begin{definition}[Measure for a Single Capacity Composite Choice]
  Given a capacity composite choice $\kappa$  we define  $\rho^{Comp}_{\cB}$ as
  \begin{small}
  \[
 \rho^{Comp}_{\cB}(\kappa) = \rho_{\bprogram}(\kappa) \widehat{\times}
  \widehat{\prod}_{\mathit{belief}(D,Elt) \in \nrml(Bel(\kappa));\cD \in doms(\kappa)}BP^{\uparrow}(D,Elt )
  \]
    \end{small}
    where $\rho_{\bprogram}$ is the measure for
     composite choices from Definition~\ref{def:comp-choice-measure}.
\end{definition}
     
The \emph{set of worlds $\omega_\kappa$ compatible with a capacity composite
choice} $\kappa$ is
\begin{small}
$$\omega_\kappa=\{w_{\sigma} = (\ruleset,\factset,\cD,\cB) \in \bworlds
  \mid Bay(\kappa)\subseteq \cF\mbox{ and } \nrml(Bel(\kappa ))\subseteq\cB\}.$$
  \end{small}
With this definition of worlds compatible with a capacity composite choice, the
definitions of equivalence of two sets of capacity composite choices,
of incompatibility of two atomic choices and of a pairwise
incompatible set of capacity composite choices are the same as for
probabilistic logic programs.
        
If $K$ is a pairwise incompatible set 
of capacity composite choices, define  $\xi_c(K)=\sum_{\kappa\in K}\rho_{\cB}^{Comp}(\kappa)$.
Given a general set $K$ of capacity composite choices, we can
construct a pairwise incompatible equivalent set through the
technique of \emph{splitting}.  In detail, (1) if $\f$ is a
probabilistic fact and $\kappa$ is a capacity composite choice that
does not contain an atomic Bayesian choice $(\f,k)$ for any
$k$, the \emph{split} of $\kappa$ on $\f$, $S_{\kappa,\f}$, is
defined as for probabilistic logic programs;  (2) if $D$ is a
domain not appearing in $doms(\kappa)$ and $B$
a \belief of $D$, the \emph{split} of
$\kappa$ on $\mathit{belief}(D,B)$, is defined as 
the set of capacity composite
choices
$S_{\kappa,D,B}=\{\kappa\cup\{\mathit{belief}(D,B)\},\kappa\cup\{\mathit{belief}(D,\neg
B)\}\}$.  In this way, $\kappa$ and $S_{\kappa,\f}$
($S_{\kappa,D,B}$) identify the same set of possible worlds,
i.e., $\omega_\kappa=\omega_{S_{\kappa,\f}}$
($\omega_\kappa=\omega_{S_{\kappa,D,B}}$)  and $S_{\kappa,\f}$ ($S_{\kappa,D,B}$)
are pairwise incompatible.
      
As for probabilistic logic programs, given a set of capacity composite choices, by
repeatedly applying splitting it is possible to obtain an
equivalent mutually incompatible set of composite
choices~\citep{DBLP:journals/jlp/Poole00} and the analogous of Theorem~\ref{incompatible_set}
can be proved.
      
\begin{theorem}
  \label{incompatible_capacity_set}
  Let $K$ be a finite set of capacity composite choices.  Then there is
  a pairwise incompatible set of capacity composite choices equivalent
  to $K$.
\end{theorem}
\begin{proof}
Given a  set of capacity composite choices $K$, there are
two possibilities to form a new set $K'$ of capacity composite choices so that
$K$ and $K'$ are equivalent:
\begin{enumerate}
\item \textbf{Removing dominated elements}: if $\kappa_1,\kappa_2\in K$ and $\kappa_1 \subset \kappa_2$, let $K' = K\setminus\{\kappa_2\}$.
\item \textbf{Splitting elements}: if $\kappa_1,\kappa_2\in K$ are compatible (and neither
is a superset of the other), there is a Bayesian choice $(\f,k)\in \kappa_1 \setminus \kappa_2$ or a belief choice $\mathit{belief}(D,B)\in \kappa_1 \setminus \kappa_2$ 
We replace $\kappa_2$ by the split of $\kappa_2$ on $\f$. Let $K' = K \setminus \{\kappa_2\} \cup S_{\kappa_2,\f}$.
\end{enumerate}
In both cases, $\omega_K=\omega_{K'}$.
If we repeat this two operations until neither is applicable, we obtain a splitting algorithm that
terminates because $K$ is finite. The resulting set $K'$ is pairwise incompatible and is equivalent to the original set.
\end{proof}

\begin{theorem}[Capacity Equivalence for Capacity Composite Choices]
	\label{thm:pairwise-capacity-equiv}
	If $K_1$ and $K_2$ are both pairwise incompatible  sets of capacity composite choices such that they are equivalent, then $\xi_c(K_1)=\xi_c(K_2)$.
\end{theorem}
\begin{proof}
The proof of this theorem is analogous to that of Theorem \ref{thm:pairwise-equiv} given by~\cite{DBLP:journals/ai/Poole93}. That proof hinges on the fact that  the probabilities of $(\f_i,0)$ and $(\f_i,1)$ sum to 1.
We must similarly prove that the belief and plausibility of $\mathit{belief}(D,B)$ and $\mathit{belief}(D,\neg B)$ sum to 1.
Let us prove it for the belief: the first component of $\rho^{Comp}_{\cB}(\{\mathit{belief}(D,B)\})$ is $\mathit{Belief}(D,B)=\sum_{X\subseteq B}mass(D,X)$
while of $\rho^{Comp}_{\cB}(\{\mathit{belief}(D,\neg B)\})$ is $\mathit{Belief}(D,\neg B)=\sum_{X\subseteq \neg B}mass(D,X)$. Clearly the two sets $\{X \mid X\subseteq B\}$ and
$\{X \mid X\subseteq \neg B\}$ partition the set of focal sets and, since the masses of focal sets sums to 1, $\mathit{Belief}(D,B)+\mathit{Belief}(D,\neg B)=1$.
For plausibility the reasoning is similar.
\end{proof}

\begin{theorem}[Normalized Capacity Space of a Program]
\label{omega_capacity_algebra}
Let $\bprogram= (\ruleset,\factset,\cD,\cB)$ be a \ourpabbv  without function symbols and let $\Omega_{\bprogram}=\{\omega_K \mid K\mbox{ is a set of capacity composite choices}\}$.
Then $\Omega_{\bprogram}=\mathbb{P}(\bworlds)$ and
$\bmu(\omega_K)=\xi_c(K')$ where $K'$ is a pairwise incompatible set of capacity composite choices equivalent to $K$.
\end{theorem}
\begin{proof}
The fact that  $\Omega_{\bprogram}\subseteq \mathbb{P}(\bworlds)$ is due to the fact that $\mathbb{P}(\bworlds)$ contains any possible set of worlds.
In the other direction, given a set of worlds $\omega=\{w_1,\ldots,w_n\}$, we can build the set of capacity composite choices $K=\{\kappa_1,\ldots,\kappa_n\}$
where $\kappa_i$ is built so that $\omega_{\kappa_i}=\{w_i\}$: let $w_i=(\cR,\cF,\cD,\cB)$, for every probabilistic fact $\f$, if $\f\in \cF$ then $(\f,1)\in 
\kappa_i$, and if  $\f\not\in \cF$ then $(\f,0)\in \kappa_i$; for every belief fact $\mathit{belief}(D,B)\in \nrml(\cB)$ then $\mathit{belief}(D,B)\in \kappa_i$.

To prove $\bmu(\omega_K)=\xi_c(K')$, it is enough to select as
$K'$ the set $K$ built as above.  
\end{proof}
The definitions of explanations and covering sets of explanations for
a query $q$ are the same as for PLP.  As a result, systems like
ProbLog and Cplint/PITA can be extended to inference over \ourpabbvs:
first a covering set of explanations is found using a conversion to
propositional logic (ProbLog~\citep{DBLP:journals/tplp/KimmigDRCR11})
or a program transformation plus tabling and answer subsumption (PITA~\citep{RigS11a}), then knowledge compilation is used to convert the
explanations into an intermediate representation (d-DNNF for ProbLog
and BDD for PITA) that makes them pairwise incompatible and allows the
computation of the belief/plausibility.
\begin{example}
  To illustrate these topics, consider {\tt r\_indep} that uses belief
  domains {\em urn1} and {\em urn2}.
  \begin{small}
    \begin{tabbing}          
      foooooooooooofofo00000000000oooooooo\=foooooooofoooooooooooooo\=ooo\=\kill
{\tt     r\_indep:-belief(urn1,\{blue\}).} \> {\tt r\_indep:- belief(urn2,\{orange\}).}
    \end{tabbing}
  \end{small}
    Since the belief domains (from Examples~\ref{ex:urn-domains} and
    \ref{ex:urn2}) are independent, a covering and pairwise incompatible set of explanations for {\tt r\_indep} is
        \begin{small}
    \begin{tabbing}          
      fooooo\=oooooooooofofooooooooofoooooooofooooooooo\=ooooo\=ooo\=\kill
      $K=\{$\>$\{\mathit{belief}(urn1,\{blue\}),\mathit{belief}(urn2,\neg\{orange\})\},$\\
       \>$\{\mathit{belief}(urn1,\neg\{blue\}),\mathit{belief}(urn2,\{orange\})\}$ \\
      \> $\{\mathit{belief}(urn1,\{blue\}),\mathit{belief}(urn2,\{orange\})\}\}$
    \end{tabbing}
    \end{small}
whose    capacity is:
    \begin{small}
     $$\xi_c(K)= ([0.1,0.7]\widehat{\times}[0.7,0.7])
      \ \widehat{+}\ ([0.3,0.9]\widehat{\times}[0.3,0.3])
      \ \widehat{+}\ ([0.1,0.7]\widehat{\times}[0.3,0.3]) = [0.19,0.97]$$
    \end{small}
Next, consider {\tt r\_dep}, both of whose rules both use the {\em urn1} belief domain
  \begin{small}
    \begin{tabbing}          
      foooooooooooofofooooooooofoooooooo\=foooooooooooooo\=ooo\=\kill
    {\tt r\_dep:- belief(urn1,\{blue\}).}\>  {\tt  r\_dep:- belief(urn1,\{red\})}
    \end{tabbing}
  \end{small}
\item
  The capacity of {\tt r\_dep} is:
  \begin{small}
    \begin{tabbing}          
      foooooooooooofofooooooooo\=foooooooofoooooooooooooo\=ooo\=\kill
      $\xi_c(\{\mathit{belief}(urn1,\{blue\}),\mathit{belief}(urn1,\{red\}) \}) = 
        \xi_c(\{\mathit{belief}(urn1,\{blue,red\})\}) = [0.4,1].$
    \end{tabbing}
  \end{small}
\end{example}

\subsubsection*{Towards a Transformation Based Implementation}
\label{sec:impl}
\newcommand{\ranvar}[1]{%
\ensuremath{\mathrm{#1}}}
\newcommand{\pitacalp}{PITAC\xspace}

We now sketch an implementation of inference in \ourpabbvs using a
PITA-like transformation~\citep{RigS11a}.  We designate as \pitacalp a
new transformation of a \ourpabbv into a normal program.  The normal
program produced contains calls for manipulating BDDs via the
\texttt{bddem} library (\url{https://github.com/friguzzi/bddem}):
\begin{itemize}
  \item $\mathit{zero/1}$, $\mathit{one/1},$
  $\mathit{and/3} $, $\mathit{or/3}$ and $\mathit{not/2}$: Boolean operations between BDDs;
  \item $get\_var\_n(R,S,\mathit{Probs},\mathit{Var})$: returns an integer indexing a variable with $|\mathit{Probs}|$ values and parameters $\mathit{Probs}$ (a list) associated to clause $R$ with
 grounding $S$;
  \item $\mathit{equality(+Var,+Value,-BDD)}$: returns a BDD 
  representing the equality $\mathit{Var} = \mathit{Value}$, i.e., that  variable with index $\mathit{Var}$ is assigned $\mathit{Value}$ in the BDD;
  \item $\mathit{mc(+BDD,-P)}$: returns the model count of the formula encoded by $\mathit{BDD}$.
  \end{itemize}

\pitacalp differs from PITA because it associates each ground atom
with three BDDs, for the probability, belief and plausibility respectively, instead of one.
The transformation for
an atom $a$ and a variable $D$, $\pitacalp(a,D)$, is $a$ with the variable
$D$  added as the last  argument. Variable $D$ will contain  triples $(PrD,BeD,PlD)$ 
with the three BDDs.

A probabilistic fact  $C_r=h:\Pi$, is transformed into the clause

\noindent
$\begin{array}{ll}
\pitacalp(C_r)=\pitacalp(h,D)\lpif get\_var\_n(r,S,[\Pi,1-\Pi],Var),  equality(Var,1,D).\\
\end{array}$

\noindent
A belief fact  $C_r=belief(\cD,B_i)$, is transformed into the clause

\noindent
$\begin{array}{ll}
\pitacalp(C_r)=\pitacalp(h,D)\lpif get\_var\_n(r,S,[\Pi_1,\ldots,\Pi_n],Var), equality(Var,i,D).\\
\end{array}$

\noindent
where the program contains the set of facts $\{mass(\cD,B_j,\Pi_j)\}_{j=1}^n$.
The transformation for clauses is the same as for PITA.

In order to answer queries, the goal {\em cap(Goal,Be,Pl)} is used,
which is defined~by

$
 \begin{array}{l}
cap(Goal,Be,Pl)\leftarrow    add\_arg(Goal,D,GoalD),\\
\ \ \ \  (call(GoalD)\rightarrow DD=(PrD,BeD,PlD),\\
\ \ \ \ \ \ mc(PrD,Pr), mc(BeD,Bel),mc(PlD,Pla),Be \ is \ Pr \cdot Bel,Pl \ is \ Pr \cdot Pla;\\
\ \ \ \ \ \  Be=0.0,Pl=0.0).
\end{array}
$

\noindent
where $add\_arg(Goal,D,GoalD)$ returns $\pitacalp(Goal,D)$ in $GoalD$.

This sketch must be refined by taking into account also the need for canonicalization and partitioning.

\section{Related Work}
\label{sec:related}
    
Among the possible semantics for probabilistic logic programs, such as CP-logic~\citep{MeeStrBlo08-ILP09-IC}, Bayesian Logic Programming~\citep{kersting2001towards}, CLP(BN)~\citep{costa2002clp}, Prolog Factor Language~\citep{DBLP:conf/ilp/GomesC12}, Stochastic Logic Programs~\citep{muggleton1996stochastic}, and ProPPR~\citep{wang2015efficient}, none of them, to the best of our knowledge, consider belief functions.

Other semantics, still not considering belief functions, target Answer Set Programming~\citep{brewka2011asp}, such as LPMLN~\citep{DBLP:conf/aaai/LeeY17}, P-log~\citep{DBLP:journals/tplp/BaralGR09}, and the credal semantics (CS)~\citep{cozman2017semantics}.
The latter deserves more attention: given an answer set program extended with probabilistic facts, the probability of a query $q$ is computed as follows.
First worlds are computed, as in the distribution semantics.
Each world $w$ now is an answer set program which may have zero or more stable models.
The atom $q$ can be present in no models of $w$, in some models, or in every model.
In the first case, $w$ makes no contribution to the probability of $q$.
If $q$ is present in \textit{some} but not all models of $w$, $P(w)$ is  computed as in Equation~\ref{def:pworld-measure}, and contributes to the \textit{upper probability} of $q$.
If instead the query is present in \textit{every} model of $w$, $P(w)$ contributes to both the upper \textit{and lower probability}.
That is, a query no longer has a point-probability but is represented by an interval.
Furthermore,~\cite{cozman2017semantics} also proved the CS is characterized by the set of all probability measures that dominate an infinitely monotone Choquet capacity.

The approach of~\cite{WanK09} incorporates belief functions into logic
programming, but
unlike ours is not based on the distribution semantics.  Rather,
belief and plausibility intervals are associated with rules. Within a
model $M$, the belief of an atom $A$ is a function of the belief intervals
associated with $A$ and the support of $A$ in $M$.

\section{Discussion} \label{sec:conclusion}

Let us now conclude the paper with a discussion about an extension of Example~\ref{example-ontology}.
\begin{example}[Reasoning with Beliefs and Probabilities]

Suppose that the UAV from Example~\ref{example-ontology} is extended
to search for stolen vehicles, alerting police when a match of
sufficient strength is detected.  The UAV can identify the type of a
vehicle using the visual model $V_{mod}$ of
Example~\ref{example-ontology}, although $V_{mod}$ alone is not
sufficient for this task.  The UAV can also detect signals from strong
RFID transponders but since it flies at a distance, it may
detect several RFID signals in the same area, leading to uncertainty.
Finally, the UAV considers only vehicles that are within a region
determined by the last reported location of the vehicle and the time
since the report.

Figure~\ref{fig:calp-prog} provides rules and pseudo-code for the
UAV's task.  
\begin{figure}[t]
\begin{small}
  \begin{tt}
\begin{tabbing}
  foooooooooooooooooo\=fooooooo\=ooofoooooooofoooooooooooooo\=ooo\=\kill
  stolen(Veh,VObj):- unusualType(Veh,Type),detection(VObj,VLoc),inArea(VLoc), \\
  \>                  belief($V_{mod}$,[VObj,Type]).\\
  stolen(Veh,VObj):- rfid(VObj,DSigl,Loc),inArea(Loc),\\
  \> targetSignal(TSig),signalMatch(DSig,TSig). \\
  n::inArea(Loc):- \>{\rm \emph{External function of last known location and time.}}\\
  n::signalMatch(Sig1,Sig2):-\>\>{\rm \emph{External function to match RFIDs.}}\\
  belief(Model,[Obj,Type]):-\>\>{\rm \emph{External function to match RFIDs.}}
\end{tabbing}
\end{tt}
\caption{UAV Rules for Stolen Vehicles.}\label{fig:calp-prog}
\end{small}
\end{figure}
The use of belief and probabilistic facts in these rules emulates an
application program's.
The first rule checks whether the make of the stolen vehicle, {\tt Veh},
is unusual.  If so, the probability that the location of the detected
object {\tt VObj} is within the search area is determined, and
combined with the belief and plausibility that {\tt VObj} is
consistent with that of {\tt Veh}.  In the second rule, if a detected
RFID signal is probabilistically in the search area, a probabilistic
match is made between the signal for {\tt VObj} and {\tt Veh}.
\end{example}
In this paper, we proposed an extension to the distribution semantics
to handle belief functions and introduced the \textit{Capacity Logic
  Programs} framework.  We precisely characterize that framework with
a set of theorems, and show how to compute the probability of queries.
Since CaLPs extend PLPs, exact inference in CaLPs is \#P-hard while
thresholded inference is PP-hard
(cf. \citep{DBLP:conf/ijcai/RaedtKT07,cozman2017semantics}), although the upper bounds for these
problems have not yet been determined. As future work, we also plan to
provide a practical implementation of our framework by extending the
cplint/PITA reasoner as outlined in Section~\ref{sec:blp-queries}.

\section*{Acknowledgements}
This work has been partially supported by Spoke 1 ``FutureHPC \& BigData'' of the Italian Research Center on High-Performance Computing, Big Data and Quantum Computing (ICSC) funded by MUR Missione 4 - Next Generation EU (NGEU).
DA and FR are members of the Gruppo Nazionale Calcolo Scientifico -- Istituto Nazionale di Alta Matematica (GNCS-INdAM).

\bibliographystyle{acmtrans}
\bibliography{biblio}

\end{document}